\documentclass[a4paper,UKenglish,cleveref, autoref]{lipics-v2019}
\usepackage[T1]{fontenc}
\usepackage[utf8]{inputenc}
\usepackage{xcolor}
\usepackage{amsmath}
\usepackage{amssymb}
\usepackage{wrapfig}
\usepackage{graphicx}
\usepackage{mathcommand}
\usepackage{thmtools}
\usepackage{thm-restate}
\nolinenumbers
\AtBeginDocument{\hypersetup{breaklinks=true,hidelinks}}

\usepackage[notion,quotation,electronic]{knowledge}

\newcommand\synonym[1]{\knowledge{#1}{synonym}}

\knowledgedirective{math notion}{notion}
\knowledgedirective{standard}{notion}

\knowledge{i.e.}{text=\textit{i.e.}}

\knowledge{ranked alphabet}{standard}

\knowledge {rank}{standard}

\knowledge{letter}[Letters|letters]{standard}
	\synonym{letter of rank~$n$}
	\synonym{letters~$a$ of rank~$n$}
	\synonym{letter~$a$ of rank~$n$}
	
\knowledge{word}[Words|words]{standard}

\knowledge {term of arity~$n$}{standard}
  \synonym {term of arity~$0$}
  \synonym {term of arity~$1$}
  \synonym {terms of arity~$n$}
  \synonym {terms of arity~$0$}
  \synonym {terms of arity~$1$}
  \synonym {term~$t$ of arity~$n$}
  
\knowledge {context}[contexts|Contexts]{standard}

\knowledge {term}[terms|Terms]{standard}

\knowledge{node}[nodes|Nodes]{standard}

\knowledge {inner nodes}{standard}

\knowledge{hole}[holes|Holes]{standard}
	\synonym {$i$th-hole}

\knowledge{root node}[root]{standard}

\knowledge{height}{standard}
	\synonym{height of a term~$s$}

\knowledge{depth}{standard}
	\synonym{depth of a context~$c$}

\knowledge{size}{standard}
	\synonym{size of a term~$s$}

\knowledge{branch}{}

\knowledge{NFA}{standard}
  \synonym{non-deterministic (tree) automaton}
  \synonym{non-deterministic automaton}
  \synonym{non-deterministic automata}
  \synonym{Non-deterministic automata}

\knowledge{run}{standard}
  \synonym{run of the automaton}
  \synonym{runs of the automaton}
  \synonym{Runs of the automaton}
  \synonym{runs}
  \synonym{Runs}

\knowledge {language accepted by the automaton}{standard}

\knowledge{accepting run}{standard}
	\synonym{accepting runs}
	\synonym{accepting run of the automaton}
	\synonym{accepting runs of the automaton}
	\synonym{Accepting runs of the automaton}
	\synonym{accepting run over it}
	\synonym{accepting runs over it}
	\synonym{accepting run~$\rho $ of the automaton}
	\synonym{accepting run of~$\maxA $ over~$c\compterm s$}
	\synonym{accepting run~$\rho'$ of~$\maxA$ over $c\compterm s$}
	\synonym{accepting run of~$\sepA $ over~$s$}
	
\knowledge{accepting run from}{standard}
	\synonym{accepting run over~$c$ from~$p$}
	\synonym{accepting run over~$c$ from~$q$}
	\synonym{accepting run over~$c$ from~$(R,P)$}
	\synonym{accepting run from~$p$}
	\synonym{accepting run~$\rho _c$ over~$c$ from some state~$p$}
	\synonym{accepting run of~$\maxA $ over~$s$}
	
\knowledge{run to}{standard}
	\synonym{run to state~$q$}
	\synonym{run~$\tau '$ over $t$ to state~$q$}
	\synonym{run over~$t$ to~$p$}
	\synonym{run over~$t$ to~$(R,P)$}
	\synonym{runs~$\rho $ of~$\maxA $ over~$s$ to a state~$p\in P$}
	\synonym{runs~$\tau $ of~$\minA $ over~$s$ to a state~$q\in P$}
	\synonym{run $\rho _t$ over~$t$ to state~$p$}
	\synonym{run over~$t$ to~$q$}
	\synonym{run of~$\biA $ over~$s$ to state~$p$}
	\synonym{run of~$\maxA $ over~$t$ to state~$p$}
	\synonym{run~$\rho ''$ over $s''$ to some state~$q$}
	\synonym{run of~$\sepA $ over~$s$ to state~$(R,P,t)$}
	\synonym{run~$\rho '$ over $t$ to state~$p$}
	\synonym{run~$\rho _s$ over~$s$ to state~$p$}
	\synonym{run of~$\biA $ over~$s$ to~$p$}
	\synonym{runs~$\tau $ of~$\biA $ over~$s$ to~$p$}
	\synonym{run~$\tau _p$ over~$s'$ to~$q_p$}
	\synonym{run~$\tau ''_p$ over~$s'$ to~$q_p$}
	\synonym{run of~$\maxA $ over~$s$ to state~$p$}
	\synonym{run over~$t=c\compterm s'$ to~$p$}
	\synonym{run~$\rho _s$ over~$s$ to~$(P,R,t)$}
	\synonym{run of~$\sepA $ over~$s$ to a state of the form~$(R,P,t)$}
	\synonym{run of $\sepA $ over~$s$ to~$(R,P,t)$}
	\synonym{run of~$\sepA $ over~$s$ to~$(R,P,t)$}
	\synonym{run~$\rho _\sep $ of~$\sepA $ over~$s$ to a state of the form~$(R,\biFinal ,t)$}
	\synonym{run $\rho '$ over~$t$ to state~$p$}
	\synonym{run of~$\maxA $ over~$a(t_0,\dots ,t_{n-1})$ to state~$p\in P$}
	\synonym{run~$\rho _i$ of~$\maxA $ over~$t_i$ to~$p_i$}
	\synonym{run~$\rho '_i$ of~$\maxA $ over~$t_i$ to~$p_i$}
	\synonym{run~$\rho '$ of~$\maxA $ over~$a(t_0,\dots ,t_{n-1})$ to~$p$}
	\synonym{runs~$\tau $ of~$\maxA $ over~$s$ to~$p$}
	\synonym{runs~$\tau $ of~$\minA $ over~$s$ to~$p$}

\knowledge{run from to}{standard}
	\synonym{run from~$p$ to~$q$}
	\synonym{run from state~$p$ to state~$q$}
	\synonym{runs of~$\maxA $ over~$m$ from~$p$ to~$p$}
	\synonym{runs of~$\minA $ over~$m$ from~$q$ to~$q$}
	\synonym{run~$\rho '$ over~$c'$ from state~$q$ to~$p$}
	\synonym{runs~$\rho $ of~$\maxA $ over~$m$ from~$p$ to~$p$}
	\synonym{runs~$\tau $ of~$\minA $ over~$m$ from~$q$ to~$q$}
	\synonym{run~$\tau '_p$ over~$m$ from~$q_p$ to~$q_p$}

\knowledge{accepted}{standard}
	\synonym{accepted by the automaton}

\knowledge{reachable}{notion}
	\synonym{reachable in~$t$}
	
\knowledge{productive}{notion}
	\synonym{productive in~$c$}
 
\knowledge{unambiguous NFA}[unambiguous automaton]{notion}

\knowledge{initial state}[Initial states|initial states|initial]{standard}

\knowledge{final state}[Final states|final states|final]{standard}
	\synonym{set of final states}
	
\knowledge{state}[states]{standard}

\knowledge{transition}[transitions|Transitions]{standard}
	\synonym{transition relation}

\knowledge{automaton with weights}{notion}

\knowledge{weight}[weights|Weights]{notion}

\knowledge{weight of a run}{notion}
	\synonym{weight of the run}
	
\knowledge{weight of an accepting run}{notion}
    \synonym{weight of the accepting run}	
	
\knowledge{tropical automaton}{notion}
   \synonym{tropical automata}
   \synonym{Tropical automata}
   \synonym{tropical (tree) automaton}

\knowledge{max-plus}[max-plus automaton|max-plus automata|Max-plus automata]{notion}

\knowledge{min-plus}[min-plus automaton|min-plus automata|Min-plus automata]{notion}

\knowledge{unambiguous tropical automaton}{notion}
	\synonym{unambiguous tropical automata}
	\synonym{Unambiguous tropical automata}

\knowledge{refines}{notion}
	\synonym{refinement with shift}
	\synonym{$t$ refines $s$ for~$P$ with shift~$x$}
	\synonym{$t$ refines~$s$ for~$P$ with shift~$x$}
	\synonym{$u$ refines~$t$ for~$P$ with shift~$y$}
	\synonym{$u$ refines~$s$ for~$P$ with shift~$x+y$}
	\synonym{$t$ refines~$s$ for~$\Prod (c)$ with shift~$x$}
	\synonym{$t$ refines $s$ for~$P$ with some shift}
	\synonym{$t$ refines~$s$ for~$P_0$ with shift~$x$}
	\synonym{refines $s$ for~$P$ with some shift}
	\synonym{$t$ refines~$a(t_0,\dots ,t_{n-1})$ for~$P$}
	\synonym{$t$ refines~$s$ for~$\Prod (c)$}
	\synonym{shift refine relation}
	\synonym{$t_i$ refines~$s_i$ for~$P_i$ with shift~$x_i$}
	\synonym{$a(t_0,\dots ,t_{n-1})$ refines~$a(s_0,\dots ,s_{n-1})$ for~$P$ with shift~$x_0+\cdots +x_{n-1}$}
	\synonym{$t$ refines~$s$ for~$P$ with shift~$\weightrun (\rho )$}
	\synonym{$t_i$ refines~$s_i$ for~$P_i$ with shift~$\weightrun (\rho _i)$}
	\synonym{$a(t_0,\dots ,t_{n-1})$ refines~$s$ for~$P$ with shift~$\weightrun (\rho _0)+\cdots +\weightrun (\rho _{n-1})$}
	\synonym{$t$ refines~$a(t_0,\dots ,t_{n-1})$ with shift~$\weight _\sep (\delta )$}
	\synonym{$t$ refines $s$ with shift~$\weightrun (\rho _0)+\cdots +\weightrun (\rho _{n-1})+\weight _\sep (\delta ) = \weightrun (\rho )$}
	\synonym{$t$ refines~$s$ for~$\biFinal $ with shift~$\weightrun (\rho _\sep )$}
	\synonym{$s$ refines~$s$ for~$P$ with shift~$0$}

\knowledge {\Q _{\max }}{math notion}
\knowledge {\Final _{\max }}{math notion}
\knowledge {\Transitions _{\max }}{math notion}
\knowledge {\weight _{\max }}{math notion}
\knowledge {\automaton _{\max }}{math notion}

\knowledge {\Q _{\min }}{math notion}
\knowledge {\Final _{\min }}{math notion}
\knowledge {\Transitions _{\min }}{math notion}
\knowledge {\weight _{\min }}{math notion}
\knowledge {\automaton _{\min }}{math notion}

\knowledge {\Q _{\sep }}{math notion}
\knowledge {\Final _{\sep }}{math notion}
\knowledge {\Transitions _{\sep }}{math notion}
\knowledge {\weight _{\sep }}{math notion}
\knowledge {\automaton _{\sep }}{math notion}

\knowledge {\Q _{\pro }}{math notion}
\knowledge {\Final _{\pro }}{math notion}
\knowledge {\Transitions _{\pro }}{math notion}
\knowledge {\automaton _{\pro }}{math notion}

\knowledge {\biA }{math notion}
\knowledge {\biFinal }{math notion}
\knowledge {\biTransitions }{math notion}
\knowledge {\biweight }{math notion}
\knowledge {\biQ }{math notion}

\knowledge {\automaton}{}
\knowledge {\Q }{}
\knowledge {\Final }{}
\knowledge {\Transitions }{}
\knowledge{\automaton_{{\max}}}{}
\knowledge{\automaton_{{\min}}}{}

\newtheorem{fact}[theorem]{Fact}
\knowledgeconfigureenvironment{theorem,lemma,proof,proposition,fact}{}

\newmathcommandPIE\sem[1]{\kl[\sem#3]{[\![}#1\kl[\sem#3]{]\!]#3}#2#4}
\knowledge {\sem}{notion}
\knowledge {\sem_{\max}}{notion}
\knowledge {\sem_{\min}}{notion}

\newrobustcmd\sep{\mathrm{sep}}
\newrobustcmd\pro{\mathrm{pro}}
\newrobustcmd\minA{\automaton_\min}
\newrobustcmd\maxA{\automaton_\max}
\newrobustcmd\sepA{\automaton_\sep}
\newrobustcmd\proA{\automaton_\pro}

\newmathcommandPIE\automaton{\kl[\automaton#2]{\mathcal{A}#2}#1#3}

\newrobustcmd\product{\cdot}

\newrobustcmd\run{\rho}
\newrobustcmd\transitions{\Delta}
\newmathcommandPIE\states{\kl[\states#1#2]{Q#1#2#3}}

\newmathcommand\powerset{\mathcal{P}}

\newcommand\Nats{\ensuremath{\mathbb{N}}}
\newcommand\Reals{\ensuremath{\mathbb{R}}}
\newrobustcmd\Ints{\mathbb{Z}}
\newrobustcmd\Rats{\mathbb{Q}}
\newrobustcmd\Comps{\mathbb{C}}

\newmathcommand\rank{\kl[\rank]{\mathrm{rank}}}
\knowledge\rank{standard}

\newmathcommand\A{A} 
\newmathcommandPIE\Q{\kl[\Q#1#2]{Q#1#2}#3} 
\newmathcommandPIE\Final{\kl[\Final#1#2]{F#1#2}#3} 
\newmathcommandPIE\Transitions{\kl[\Transitions#1#2]{\Delta#1#2}#3} 

\newmathcommand\biA{\kl[\biA]{\mathcal{A}}}
\newmathcommand\biFinal{\kl[\biFinal]{F}}
\newmathcommand\biTransitions{\kl[\biTransitions]{\Delta}}
\newmathcommand\biQ{\kl[\biQ]{Q}}
\newmathcommand\biweight{\kl[\biweight]{\mathrm{weight}}}

\newmathcommandPIE\weight{\kl[\weight#2]{\mathrm{weight}#2}#1#3} 
\newmathcommandPIE\weightrun{\kl[\weightrun#2]{\mathrm{weight}#2}#1#3} 
\newmathcommandPIE\weightacc{\kl[\weightacc#2]{\mathrm{weight}^{\mathrm{acc}}#2}#1#3} 
\knowledge\weight{math notion}
\knowledge\weightrun{math notion}
\knowledge\weightacc{math notion}

\ExplSyntaxOn
\newcommand\newTermMathNotion[3][math~notion]{%
  \newmathcommand#2{\kl[#2]{\mathrm{#3}}}%
  \exp_args:Nc\newmathcommand{intro\cs_to_str:N#2}{\intro[#1]{\mathrm{#3}}}
  \knowledge#2{#1}%
  }
\ExplSyntaxOff

\newTermMathNotion[standard]\Nodes{Nodes}
\newTermMathNotion[standard]\Terms{Terms}
\newTermMathNotion[standard]\Contexts{Contexts}
\newTermMathNotion[standard]\height{height}
\newTermMathNotion[standard]\depth{depth}
\newTermMathNotion[standard]\size{size}

\newTermMathNotion[notion]\Prod{Prod}
\newTermMathNotion[notion]\Reach{Reach}
\newTermMathNotion[notion]\Reachable{Reachable}
\newTermMathNotion[notion]\Productive{Productive}
\newTermMathNotion[notion]\sr{sr}

\renewmathcommand\k{k}

\renewmathcommand\root{\kl[\root]{\mathrm{root}}}
\knowledge\root{standard}

\newcommand\compterm{\mathrel{\kl[\compterm]{\circ}}}
\knowledge\compterm{standard}

\title{Unambiguous separators for tropical tree automata}
\author{Thomas Colcombet}{IRIF, CNRS}{thomas.colcombet@irif.fr}{}
		{Supported by the European Research Council (ERC) under the European Union’s Horizon 2020 research and innovation programme (grant agreement No.670624), and the DeLTA ANR project (ANR-16-CE40-0007)}
\author{Sylvain Lombardy}{LaBri, Université de Bordeaux}{sylvain.lombardy@labri.fr}{}{}
\authorrunning{T. Colcombet and S. Lombardy}


\begin{document}

\maketitle

\begin{abstract}
In this paper we show that given a max-plus automaton (over trees, and with real weights) computing a function~$f$ and a min-plus automaton (similar) computing a function~$g$ such that~$f\leqslant g$, there exists effectively an unambiguous tropical automaton computing~$h$ such that~$f\leqslant h\leqslant g$.

This generalizes a result of Lombardy and Mairesse of 2006 stating that series which are both max-plus and min-plus rational are unambiguous. This generalization goes in two directions: trees are considered instead of words, and separation is established instead of characterization (separation implies characterization). The techniques in the two proofs are very different.
\end{abstract}

\section{Introduction}

"Tropical automata" is a nickname for weighted automata (automata parameterized by a semiring as introduced by Schützenbgerger \cite{Schutzenberger61b}) over a tropical semiring.
This is a particularly simple model of finite state automata that describe functions rather than languages.
It exists in two forms, "max-plus" and "min-plus automata".
Essentially, a tropical automaton~$\automaton$ is a  non-deterministic automaton for which each transition is labelled by a real weight (or an integer, or a natural number, depending on the variants). 
This weight is extended into a weight for a run: the sum of the weights of the transitions involved. 
A "max-plus automaton" computes the function $\sem{\automaton}\colon A^*\rightarrow \Reals\cup\{\bot\}$
which to an input word associates the maximum weight of an accepting run over the input, or~$\bot$ if there is no accepting runs. If it is a "min-plus automaton",  minimum is used instead of maximum.

The use of tropical automata arises naturally in different contexts:
"max-plus automata" have been used for modeling scheduling constraints
(see for instance \cite{GaubertMairesse95}) or worst case behaviors (see for instance \cite{ColcombetDaviaudZuleger17} for computing the asymptotic worst case execution time of loops under the size-change abstraction); 
"min-plus automata" are used for optimisation questions (these are for instance used as a key tool in the decision of the star-height problem \cite{Hashiguchi88}). In all these situations, non-trivial decision procedures are used (\cite{Hashiguchi82,Leung91,ColcombetDaviaudZuleger14}).

The starting point of this work is a result from 2006 of Lombardy and Mairesse:
    \begin{theorem}[\cite{LombardyMairesse06,LombardyMairesse07}]\label{theorem:lombardy-mairesse}
    A map~$f\colon A^*\rightarrow \Reals+\{\bot\}$ 
    which is both definable by a min-plus and by a max-plus  automaton is definable by an unambiguous tropical automaton.
	\end{theorem}
Recall that an automaton is "unambiguous@unambiguous NFA" if there is at most one accepting run
per input\footnote{Note that when a tropical automaton is unambiguous,
       it makes no difference whether it is a max-plus or a min-plus automaton:
       It computes the same function.}.
Unambiguous automata form a very particular class of tropical automata. Most of the problems which are open or undecidable for general tropical automata are easily decidable for unambiguous automata: equivalence with another tropical automaton \cite{Krob94}, boundedness, sequentiality, description of the asymptotic behaviour \cite{ColombetDaviaud16}.  

It is noteworthy that the decision of sequentiality actually applies to unambiguous automata and that algorithms described for larger classes (finitely or polynomially ambiguous), consist indeed in deciding first whether the 
tropical automaton is equivalent to some unambiguous one \cite{KlimannLombardyMairessePrieur04,KirstenLombardy09}.

The above \cref{theorem:lombardy-mairesse} belongs to a fascinating corpus of mathematical statements of the form `if $X$ belongs both to class~$C$ and to class~$D$, then it belongs to class~$E$', where~$E$ is structurally simpler than both~$C$ and~$D$ (often $D$ is some form of dual of~$C$). An archetypical example arises in descriptive set theory: Suslin's theorem states that 
    \begin{quote}
    if a set is analytic and coanalytic, it is Borel. 
    \end{quote}
Many other instances of this pattern exist. For instance in automata theory, if an infinite tree language is Büchi and its complement is Büchi, it is weak (Rabin's theorem \cite{Rabin70}). This extends to cost-functions over infinite trees: if a cost-function over infinite trees is both B-Büchi and S-Büchi, it is quasi-weak ; over infinite words, it is even weak (Kuperberg and Vanden Boom \cite{KuperbergBoom11,KuperbergBoom12}). For languages of infinite words beyond regular, if a language is $\omega B$ and $\omega S$ definable, then it is $\omega$-regular (Skrzypczak \cite{Skrzypczak14}). In language theory, a language which is both $\Sigma_2$ and $\Pi_2$ definable is definable in the two variables fragment (Thérien and Wilke \cite{TherienWilke04}). 
Also, a language which is both the support and the complement of the support of a rational series over a field is regular \cite{RestivoReutenauer84}.
This list continues on and on.

In many situations such statements arise in fact from a more general result of `separation' (or of `interpolation' in the logical terminology).   For instance, Lusin's theorem is the separation version of Suslin's theorem: It states that
   \begin{quote}
   for $X\subseteq Y$ with $X$ analytic and $Y$ coanalytic, then $X\subseteq Z\subseteq Y$ for some Borel set~$Z$.
   \end{quote}
Such separation results imply the characterization version. For instance, Suslin's result follows from Lusin's theorem: take $X=Y$ to be the set which is both analytic and coanalytic.
Then $X\subseteq Z\subseteq Y=X$ for $Z$ Borel; hence $X$ is Borel. This relationship is general.
The results of Rabin, Vanden Boom and Kuperberg, and Skrzypczak, for instance, exist in a `separation variant'.

\paragraph*{Contribution}
The natural question that we answer in this work is thus:
  \begin{quote}
  Does there exist a separation version of \cref{theorem:lombardy-mairesse} ?
  \end{quote}
In this paper, we provide a positive answer to this question. It takes the following form:
  \begin{restatable}[separation for tropical tree automata]{theorem}{StateMainTheorem}%
      \label{theorem:main}
      Given a "max-plus automaton"~$\maxA$ and a "min-plus automaton"~$\minA$
      such that\footnote{In this statement, we assume
          that $\bot$ is incomparable with other elements, and thus
          $\sem\maxA$ and $\sem\minA$ have same support.}
      \begin{align*}
      \sem\maxA\leqslant\sem\minA\ ,
      \end{align*}
      there exists effectively an "unambiguous tropical automaton"~$\sepA$ such that
      \begin{align*}
      \sem\maxA\leqslant\sem\sepA\leqslant\sem\minA\ .
      \end{align*}
  \end{restatable}
Let us stress that the above theorem is established in the  context of tropical automata over trees.
\cref{theorem:lombardy-mairesse} is now a corollary. Indeed, (a) tropical word automata are a particular case of tree automata over a ranked alphabet made of unary symbols only, plus a constant, and (b) assuming that~$f$ is both accepted by a "min-plus" and by a "max-plus automaton", then by \cref{theorem:main}, there exists a function~$h$ accepted by an "unambiguous tropical automaton" such that $f\leqslant h\leqslant f$. Thus~$f=h$ is accepted by an "unambiguous tropical automaton".

Note that, though the result is a generalization, the proof of \cref{theorem:main} is very different from the original one of \cref{theorem:lombardy-mairesse}.

Let us finally emphasize that particular care has been taken in order to obtain the result for real weights. Indeed, in the integer case (and as a consequence in the rational case), simpler techniques can be used that involve keeping in the finitely many states of the result automaton some explicit differences of partial weights up to a certain bound. Such a technique (as far as we know) cannot be used in the real case. 

\paragraph*{Other related work}
The class of unambiguous tropical automata form an interesting subclass of tropical automata.
In particular, equivalence is decidable, while the problem for max-plus or min-plus automata is undecidable \cite{Krob94}. Given a tropical automaton, deciding unambiguity is an open problem. It has been solved if the input automaton is finitely ambiguous in \cite{KlimannLombardyMairessePrieur04},  and when it is polynomially ambiguous in \cite{KirstenLombardy09}.

\paragraph*{Structure of the paper}
This paper is organized as follows.
In \cref{section:definitions}, we recall the standard definitions concerning trees, automata over trees, and tropical automata.
In \cref{section:main-theorem}, we establish our main theorem of separation, \cref{theorem:main}.
\Cref{section:conclusion} concludes.

\section{Definitions}
\label{section:definitions}

We review in this section classical notions concerning terms, and then automata and
tropical automata.\footnote{Note that this document incorporates
	internal links relating each notion used to its introduction. These have been generated using  
	the \LaTeX\ package \texttt{knowledge}, and are 
	usable in all PDF viewers (some do even offer an overview of the definition when hovering above a word).
	Any feedback welcome.}

\subsection{Terms and contexts}

A ""ranked alphabet"" is a set~$\A$, the elements of which are called ""letters"",
together with a map $\intro*\rank$ from~$\A$ to~\Nats. \phantomintro{rank}%
\AP
For~$n\in\Nats$, let~$\intro*\Terms(n)$ be the set of ""terms of arity~$n$"" over the alphabet $\A\cup\{1,\dots,n\}$ in which~$1,\dots,n$ are seen as special letters of rank~$0$ that are used exactly once in each term. 
We call simply ""terms"" the "terms of arity~$0$", and the set of "terms" is simply denoted~$\reintro*\Terms$.
\AP We call ""context"" the "terms of arity~$1$", and the set of "contexts" is simply denoted~$\intro*\Contexts$. Note that each "letter"~$a$ of "rank"~$r$ can naturally be seen as a "term of arity~$n$" consisting solely of a root labelled~$a$ and children~$1,\dots,n$.
\AP
The ""nodes"" of a "term of arity~$n$", $\intro*\Nodes(t)$ is the set of positions of the letters in the term.
The ""root node"" is denoted~$\intro*\root$. A node labelled~$i$ for~$i=1\dots n$ is called the ""$i$th-hole"". The "nodes" that are not "holes" are called ""inner nodes"".
Given a "node"~$x\in\Nodes(t)$, $t(x)$ denotes the letter it carries. Given a "letter of rank~$n$" and "terms" $t_0,\dots,t_{n-1}$, we denote by~$a(t_0,\dots,t_{n-1})$ the "term" that has~$a$ as "root", and as children from left to right~$t_0,\dots,t_{n-1}$.
\AP
The ""height of a term~$s$"", denoted~$\intro*\height(s)$, is the longest length of a "branch", for the standard meaning of a branch.
\AP
The ""size of a term~$s$"", denoted~$\intro*\size(t)$, is the number of "nodes" it has.
\AP Finally, given a "context"~$c$ and~$t$ a "term" (resp. $t$ another "context"),
we denote~$c\intro*\compterm t$ the "term" (resp. the "context") obtained by plugging the "root" of~$t$ in the "hole" of~$c$.

\subsection{Automata}

\AP A ""non-deterministic (tree) automaton"" (or simply an "automaton@NFA") has a finite set of ""states""~$\Q$, an input ranked alphabet~$\A$, a ""set of final states""~$\Final$, and a ""transition relation""~$\Transitions$ that consists of tuples of the form $(p_0,\dots,p_{n-1},a,q)$
in which~$a\in\A$ is a "letter of rank~$n$", and~$p_0,\dots,p_{k-1},q$ are "states" from~$\Q$.

\AP
A ""run of the automaton"" over a "term~$t$ of arity~$n$" is a
map~$\rho$ from $\Nodes(t)$ to~$\Q$  such that for all "inner nodes"~$x\in\Nodes(t)$ of
children~$x_0,\dots,x_{n-1}$, $(\rho(x_0),\dots,\rho(x_{n-1}),t(x),\rho(x))\in\Transitions$.
We shall write~$\tilde\rho(x)$ this "transition".
\AP An ""accepting run"" is a "run of the automaton" such that~$\rho(\root)\in\Final$. Given a "term"~$t$, $t$ is ""accepted by the automaton"" if there exists an "accepting run of the automaton" over~$t$.
The set of "terms" that are accepted is the ""language accepted by the automaton"".
We slightly refine the terminology for easier use.
\AP
Over a "term", a ""run to state~$q$"" is a run that assumes "state"~$q$ at the "root".
Over a "context", a ""run from state~$p$ to state~$q$"" signifies that the state assumed in the "hole" is~$p$, and the one assumed at the "root" is~$q$. An ""accepting run from~$p$"" is a "run from~$p$ to~$q$" for a "final state"~$q$.

\AP
An "automaton@NFA" is ""unambiguous@unambiguous NFA"" if for all input "terms"~$t$, there exists at most one "accepting run over it".
Said differently, for all input "terms"~$t$, either there are no "accepting runs over it", and the term is not "accepted",
or there is exactly one "accepting run", and the "term" is "accepted".

\AP
An ""automaton with weights""\footnote{This is not the a weighted automaton,
      which is parametrized by a semiring and not a monoid.
      This definition serves here just for holding the structure
      of our tropical automata irrespective of whether these are
      "min-plus" or "max-plus".}
$\automaton$ is a "non-deterministic automaton" together with a real ""weight""
for all "transitions" and all "final states", i.e. a map~$\intro*\weight$ from~$\Transitions\uplus\Final$ to~$\Reals$.
\AP 
Given a "run~$\rho$ of the automaton@run", the ""weight of the run""~$\intro*\weightrun(\rho)$ is the sum of the "weights" of~$\tilde\rho(x)$
for~$x$ ranging over the "inner nodes" of~$t$.
\AP
Given an "accepting run~$\rho$ of the automaton",
the ""weight of the accepting run""~$\intro*\weightacc(\rho)$ is the sum of the "weight of the run" and $\weight(\rho(\root))$.

\AP ""Tropical automata"" refer in this work to one of two forms of automata: "min-plus automata" and "max-plus automata" defined as follows. 
\AP A ""min-plus automaton""~$\automaton$ is an "automaton with weights" that computes a function:
\AP\phantomintro{\sem_{\min}}
\begin{align*}
 \sem{\automaton}_\min\colon \Terms&\longrightarrow \Reals\uplus\{\bot\} \\
 						 t &\longmapsto \begin{cases}
						 	&\bot\qquad\qquad\quad\qquad\text{if there are no "accepting runs" of $\automaton$ over~$t$,}\\
							&\min\{\weightacc(\rho)\mid\rho\text{ "accepting run" of~$\automaton$ over~$t$}\} \quad\text{otherwise,}
						 	\end{cases}
\end{align*}
in which~$\bot$ is a symbol that we understand as `undefined' (it appears classically as an absorbing element for~$+$ which is larger than all~$x\in\Reals$, i.e., the zero of the tropical semiring).
\AP A ""max-plus automaton"" is defined in an identical manner, but the semantics \phantomintro{\sem_{\max}}$\sem\automaton_\max$ is now defined using $\max$ instead of $\min$. Since it is always clear from the context, we denote simply by \phantomintro{\sem}$\sem\automaton$ either
$\sem\automaton_\min$ or $\sem\automaton_\max$ depending on whether~$\automaton$ has been declared as a "min-plus@min-plus automaton" or as a "max-plus automaton".

\AP An ""unambiguous tropical automaton""~$\automaton$ is a "tropical automaton" that has an "unambiguous@unambiguous NFA" underlying "automaton@NFA".
Note that in this case, $\sem\automaton_\max=\sem\automaton_\min$, and hence we call it simply "tropical automaton" and do not have to specify whether it is a "min-plus@min-plus automaton" or "max-plus@max-plus automaton".

\section{Separating "tropical automata"}
\label{section:main-theorem}

\subsection{Statement and structure of the proof}

The goal of this section is to prove our main theorem:
\StateMainTheorem*
\AP
From now on, we fix the "ranked alphabet"~$\A$,
a "max-plus automaton"~$\maxA$ and a "min-plus automaton"~$\minA$:%
\phantomintro{\automaton_{\min}}\phantomintro{\automaton_{\max}}\phantomintro{\weight_{\min}}\phantomintro{\weight_{\max}}\phantomintro{\Q_{\min}}\phantomintro{\Q_{\max}}\phantomintro{\Final_{\min}}\phantomintro{\Final_{\max}}\phantomintro{\Transitions_{\min}}\phantomintro{\Transitions_{\max}}
\begin{align*}
\reintro*\maxA&=(\reintro*\Q_\max,\A,\reintro*\Final_\max,\reintro*\Transitions_\max,\reintro*\weight_\max)&\!\!\text{and}\!\!\!\quad
\reintro*\minA&=(\reintro*\Q_\min,\A,\reintro*\Final_\min,\reintro*\Transitions_\min,\reintro*\weight_\min)
\end{align*}
such that
\begin{align*}
	\sem{\minA}\leqslant\sem\maxA\ .
\end{align*}
\AP
It will be convenient in what follows to consider a single "automaton with weights"~$\intro*\biA$
constructed as the disjoint union of~$\maxA$ and~$\minA$
(of course, it should be neither seen as a "min-plus automaton" nor as a "max-plus automaton").
Formally, we assume without loss of generality that~$\Q_\max$ and~$\Q_\min$ are disjoint, and we set
\AP \phantomintro\biFinal\phantomintro\biTransitions\phantomintro\biweight\phantomintro\biQ
\begin{align*}
  \reintro*\biQ&=\Q_\min\cup\Q_\max, \qquad
  \reintro*\biFinal=\Final_\max\cup\Final_\min,\qquad
  \reintro*\biTransitions=\Transitions_\max\cup\Transitions_\min,\\
  &\text{and}\quad\reintro*\biweight(v)=\begin{cases}
  			\weight_\max(v)&\text{for~$v\in\Transitions_\max\uplus\Final_\max$}\\
			\weight_\min(v)&\text{otherwise.}
			\end{cases}
\end{align*}
The rest of this section is devoted to the proof of \cref{theorem:main}, and is organized as follows.
In \cref{section:reachable-productive}, we use some classical automata constructions
for accessing in an unambiguous manner the "reachable" and "productive" states (\cref{lemma:lookahead}).
The combinatorial core of the proof is contained in \cref{section:pumping} in which we study
how the values of the automata may evolve in a context (\cref{lemma:pumping}), and use it 
for showing how terms can be substituted while preserving separability (\cref{corollary:always-refinable}).
We finally provide the construction of the automaton~$\sepA$ in  \cref{section:construction}, and establish its correctness (\cref{lemma:correctness}). This concludes the proof of \cref{theorem:main}.

\subsection{Reachable and productive states}
\label{section:reachable-productive}

An ingredient which is necessary in the proof is that the automaton we construct is always `aware' of what are the
states that may lead to an accepting run to the root. This section is concerned with this aspect, and involves only completely standard techniques for tree automata.

\AP
Given a "term"~$t$, set~$\intro*\Reach(t)\subseteq\biQ$ to be the set of "states"~$p$ such that there
is a "run over~$t$ to~$p$". We call such states ""reachable in~$t$"".
Given a "context"~$c$, set~$\intro*\Prod(c)\subseteq\biQ$ to be the set of "states"~$p$ such that there
is an "accepting run from~$p$". We call such states ""productive in~$c$"".
\AP We finally set \phantomintro\Reachable\phantomintro\Productive
\begin{align*}
  \reintro*\Reachable &= \{\Reach(c)\mid c\in\Contexts\}&\text{and}\quad
  \reintro*\Productive &= \{\Prod(t)\mid t\in\Terms\}.
\end{align*}

\AP We now construct an "automaton@NFA"~$\intro*\proA=(\Q_\pro,\A,\Final_\pro,\Transitions_\pro)$ that computes the productive states at each node of a term. \AP The "states" are~$\intro*\Q_\pro=\Reachable\times\Productive$. \AP The "final states" $\intro*\Final_\pro=\Reachable\times\{\biFinal\}$, and \AP
for all "letters~$a$ of rank~$n$", the automaton has a "transition" of the form\phantomintro{\Transitions_{\pro}}
  \begin{align*}
     ((R_0,P_0),\dots,(R_{n-1},P_{n-1}),a,(R,P))\in\Transitions_\pro
  \end{align*}
  whenever
  \begin{itemize}
  \item $R=\{r\in\biQ\mid (r_0,\dots,r_{n-1},a,r)\in\biTransitions,~r_j\in R_j\text{ for all~$j$}\}$, and
  \item $P_i=\{ r_i\in\biQ \mid (r_0,\dots,r_{n-1},a,p)\in\biTransitions,~r_j\in R_j\text{ for $j\neq i$},~p\in P\}$ for   all~$i=0\dots n-1$.
\end{itemize}
In the above definition, the constraint on~$R$ induces the computation in a bottom-up deterministic way of the set
of "states" that are "reachable" from the term below.  The constraint on~$P_i$ computes similarly
in a top-down deterministic way the set of "states" that are "productive" in the context above. We do not prove the correctness of this construction further. The importants aspects of this construction are summarized in the following lemma.
\begin{lemma}\label{lemma:lookahead}
  For all~$P\in\Productive$ and all "terms"~$t$,
  there exists one and one only "run of~$\proA$ over~$t$ to a state of the form~$(R,P)$@run to" for some~$R\in\Reachable$.
  And furthermore, $R=\Reach(t)$.
  
  For all~$R\in\Reachable$ and all "contexts"~$c$,
  there exists one and only one "accepting run of~$\proA$ over~$c$ from a state of the form~$(R,P)$@accepting run from"
  for some~$P\in\Productive$.
  And furthermore, $P=\Prod(c)$.
\end{lemma}

\subsection{The central pumping lemma}
\label{section:pumping}

In this section, we establish the key~\cref{corollary:always-refinable}. The central concept here is to understand what it does for the value computed by~$\maxA$ and by~$\minA$ to substitute a subtree for another subtree. And more precisely, we devise sufficient conditions such that, after performing the substitution, the values of the two automata gets closer one to the other, up to some shifting. This property is expressed in \cref{lemma:shift-refine-improves}.

The key definition involved is the one of "refinement with shift" as defined now.
\begin{definition}\AP
  Given two "terms"~$s,t$, some set~$P\subseteq\biQ$, and some real number~$x$, then ""$t$ refines $s$ for~$P$ with shift~$x$"" if 
  \begin{itemize}
  \item $\Reach(s)=\Reach(t)$,
  \item for all "runs~$\rho$ of~$\maxA$ over~$s$ to a state~$p\in P$", there is
  		a "run~$\rho'$ over $t$ to state~$p$" such that
		\begin{align*}
		  \weightrun(\rho)\leqslant\weightrun(\rho')+x
		\end{align*}
		and
  \item for all "runs~$\tau$ of~$\minA$ over~$s$ to a state~$q\in P$", there is
  		a "run~$\tau'$ over $t$ to state~$q$" such that
		\begin{align*}
		  \weightrun(\tau')+x\leqslant\weightrun(\tau)
		\end{align*}
  \end{itemize}
\end{definition}

The justification of this definition is given by the following lemma. It shows how substituting~$s$ for~$t$ in a context
when "$t$ refines~$s$ with some shift@refinement with shift" is done while `staying in the separation interval'.
\begin{lemma}\label{lemma:shift-refine-improves}
  Let~$c$ be a "context", and~$s,t$ be "terms" such that "$t$ refines~$s$ for~$\Prod(c)$ with shift~$x$", then
  \begin{align*}
     \sem\maxA(c\compterm s)\leqslant \sem\maxA(c\compterm t) + x \leqslant \sem\minA(c\compterm t) +x \leqslant\sem\minA(c\compterm s)\ .
  \end{align*}
\end{lemma}
\begin{proof}
   Let~$\rho$ be an "accepting run of~$\maxA$ over~$c\compterm s$".
   It can be decomposed as an "accepting run~$\rho_c$ over~$c$
   from some state~$p$" and a "run~$\rho_s$  over~$s$ to state~$p$".
   The run $\rho_c$ is a witness that~$p\in\Prod(c)\cap\Q_\max$. Hence, since
   "$t$ refines~$s$ for~$\Prod(c)$ with shift~$x$", there exists a "run $\rho_t$ over~$t$ to state~$p$"
   such that~$\weightrun(\rho_s)\leqslant\weightrun(\rho_t)+x$. By gluing~$\rho_t$ with~$\rho_c$,
   we obtain a new "accepting run~$\rho'$ of~$\maxA$ over $c\compterm s$", and furthermore,
   \begin{multline*}
       \weightacc(\rho)=\weightacc(\rho_c)+ \weightrun(\rho_s)\\
            \leqslant\weightacc(\rho_c)+\weightrun(\rho_t)+x=\weightacc(\rho')+x\ .
   \end{multline*}
   Since for all~$\rho$ there exists such a~$\rho'$, we obtain
   \begin{align*}
      \sem\maxA(c\compterm s)\leqslant\sem\maxA(c\compterm t)+x\ .
   \end{align*}

   The middle inequality simply comes from the key assumption~$\sem\maxA\leqslant\sem\minA$ in \cref{theorem:main}.
  
   The third inequality is established as the first one (it is in symmetric).
\end{proof}
The two following facts are straightforward to verify.
\begin{fact}[reflexivity of refinement with shift]\label{fact:shift-refine-reflexivity}
  For all "terms"~$s$, and all~$P\subseteq\biQ$, "$s$ refines~$s$ for~$P$ with shift~$0$".
\end{fact}
\begin{fact}[transitivity of refinement with shift]\label{fact:shift-refine-transitivity}
  If "$t$ refines~$s$ for~$P$ with shift~$x$", and~"$u$ refines~$t$ for~$P$ with shift~$y$",
  then "$u$ refines~$s$ for~$P$ with shift~$x+y$".
\end{fact}
The next lemma is also purely mechanical.
\begin{lemma}[refinement with shift is a congruence]\label{lemma:shift-refine-congruence}
  Let~$((P_0,R_0),\dots,(P_{n-1},R_{n-1}),a,(P,R))\in \Transitions_\pro$,
  and for all~$i=0\dots n-1$, let $t_i,s_i$ be "terms" such that
  \begin{itemize}
  \item $\Reach(t_i)=R_i$, and
  \item "$t_i$ refines~$s_i$ for~$P_i$ with shift~$x_i$",
  \end{itemize}
  then "$a(t_0,\dots,t_{n-1})$ refines~$a(s_0,\dots,s_{n-1})$ for~$P$ with shift~$x_0+\cdots+x_{n-1}$".
\end{lemma}
\begin{proof}
   Let~$\rho$ be a "run of~$\maxA$ over~$a(t_0,\dots,t_{n-1})$ to state~$p\in P$".
   The run~$\rho$ can be decomposed into a "transition"~$(p_0,\dots,p_{n-1},a,p)$ of "weight"~$x$ at the root,
   and a "run~$\rho_i$ of~$\maxA$ over~$t_i$ to~$p_i$" for all~$i=0\dots n-1$.
   For all~$i=0\dots n-1$, $p_i\in\Reach(t_i)=R_i$.
   Since furthermore~$p\in P$ and $((P_0,R_0),\dots,(P_{n-1},R_{n-1}),a,(P,R))\in\Transitions_\pro$,
   we obtain that~$p_i\in P_i$ for all~$i=0\dots n-1$ (second item of the definition of~$\Transitions_\pro$).
   Thus, since "$t_i$ refines~$s_i$ for~$P_i$ with shift~$x_i$",
   there exists a "run~$\rho'_i$ of~$\maxA$ over~$t_i$ to~$p_i$"
   such that~$\weight(\rho_i)\leqslant\weight(\rho'_i) + x_i$.
   We can combine all these runs~$\rho'_i$ together with the transition~$(p_0,\dots,p_{n-1},a,p)$ and obtain
   a new "run~$\rho'$ of~$\maxA$ over~$a(t_0,\dots,t_{n-1})$ to~$p$" such that
   \begin{align*}
       \weight(\rho)&=\weight(\rho_0)+\dots+\weight(\rho_{n-1}) + x\\
       					&\leqslant\weight(\rho'_0)+x_0+\dots+\weight(\rho'_{n-1})+x_{n-1}+x\\
						&=\weight(\rho') + x_0 + \dots + x_{n-1}\ .
   \end{align*}
   This shows half of the fact that "$a(t_0,\dots,t_{n-1})$ refines~$a(s_0,\dots,s_{n-1})$ for~$P$ with shift~$x_0+\cdots+x_{n-1}$". The other half is symmetric.
\end{proof}

We aim now at proving \cref{corollary:always-refinable} which states that all sufficiently large term
is `"shift refined@refinement with shift"'
by another one of smaller size. Beforehand, we need a pumping argument for establishing:
\begin{lemma}\label{lemma:pumping}
    Let~$P\in\Productive$, $R\in\Reachable$ and~$m$ be a "context", then there exists a real number~$x$
    such that:
    \begin{itemize}
    \item for all "runs~$\rho$ of~$\maxA$ over~$m$ from~$p$ to~$p$" with~$p\in P\cap R$, $\weightrun(\rho)\leqslant x$, and
    \item for all "runs~$\tau$ of~$\minA$ over~$m$ from~$q$ to~$q$" with~$q\in P\cap R$, $x\leqslant\weightrun(\tau)$.     
    \end{itemize}
\end{lemma}
\begin{proof}
  Let~$t$ be a "term" such that~$\Reach(t)=R$, and $c$ be a "context" such that~$\Prod(c)=P$.
  
  \noindent\textbf{$\vartriangleright$ Claim:} We claim first that for all "runs~$\rho$ of~$\maxA$ over~$m$ from~$p$ to~$p$" with~$p\in P\cap R$
  and all "runs~$\tau$ of~$\minA$ over~$m$ from~$q$ to~$q$" with~$q\in P\cap R$, $\weightrun(\rho)\leqslant \weightrun(\tau)$.
  \AP
  Indeed, otherwise, there would exist some runs~$\rho,\tau$ as above such that~$\weightrun(\rho)>\weightrun(\tau)$.
  \AP I.e.
  \begin{align}
   \weightrun(\tau)-\weightrun(\rho)<0\ . \tag{$\star$}\label{eq-lemma-pumping}
  \end{align}
  Consider now for all~$n>0$ the  term:
  \begin{align*}
      u_n = c\compterm \overbrace{m\compterm \cdots \compterm m}^{\text{$n$-times}}\compterm t\ .
  \end{align*}
  Let~$\rho'$ be some "accepting run over~$c$ from~$p$" (this is possible since~$p\in P=\Prod(c)$).
  Let~$\tau'$ be some "accepting run over~$c$ from~$q$" (this is possible since~$q\in P=\Prod(c)$).
  Let~$\rho''$ be some "run over~$t$ to~$p$" (this is possible since~$p\in R=\Reach(t)$).
  Let~$\tau''$ be some "run over~$t$ to~$q$" (this is possible since~$q\in R=\Reach(t)$).
  
  By concatenating~$\rho'$, $n$-times $\rho$, and~$\rho''$, we obtain an "accepting run"~$\rho_n$ over~$u_n$
  of weight~$\weightacc(\rho_n)=\weightacc(\rho')+n\,\weightrun(\rho)+\weightrun(\rho'')$.
  Similarly, by concatenating~$\tau'$, $n$-times $\tau$, and~$\tau''$, we obtain an "accepting run"~$\tau_n$ over~$u_n$
  of weight~$\weightacc(\tau_n)=\weightacc(\tau')+n\,\weightrun(\tau)+\weightrun(\tau'')$.

  Furthermore, since~$\sem\maxA(u_n)\leqslant\sem\minA(u_n)$,
  $\weightacc(\rho_n)\leqslant\weightacc(\tau_n)$. We obtain
  \begin{align*}
     0&\leqslant\weightacc(\tau_n)-\weightacc(\rho_n)\\
       &=\weightacc(\tau')+n\,\weightrun(\tau)+\weightrun(\tau'')-\weightacc(\rho')-n\,\weightrun(\rho)-\weightrun(\rho'')\\
       &=(\weightacc(\tau') + \weightrun(\tau'') -\weightacc(\rho') -\weightrun(\rho'')) + n(\weightrun(\tau)-\weightrun(\rho))\ .
  \end{align*}
  However, using \eqref{eq-lemma-pumping}, this last quantity tends to $-\infty$ when~$n$ tends to~$\infty$. It contradicts  its non-negativeness. The claim is established.
  \medskip
  
  We can now establish the lemma. Let~$Y$ be the set of weights~$\weight(\rho)$ for~$\rho$
  ranging over the "runs of~$\maxA$ over~$m$ from~$p$ to~$p$" with~$p\in P\cap R$.
  Similarly, let~$Z$ be the set of weights~$\weight(\tau)$ for~$\tau$
  ranging over the "runs of~$\minA$ over~$m$ from~$q$ to~$q$" with~$q\in P\cap R$.
  The above claim states that for all~$y\in Y$ and all~$z\in Z$, $y\leqslant z$.
  This implies the existence of some real number~$x$ such that for all~$y\in Y$, $y\leqslant x$,
  and for all~$z\in Z$, $x\leqslant z$
  (note that proving it requires to treat the case of~$Y$ and/or~$Z$ being empty, and thus requires a case distinction).
  This is exactly the statement of the lemma.
\end{proof}

\begin{lemma}\label{lemma:always-refinable}
There exists a computable~$\k\in\Nats$ such that for all $P_0\in\Reachable$ and all "terms"~$s$ of "height" more than~$\k$,
there exists effectively a "term"~$t$ such that "$t$ refines $s$ for~$P$ with some shift" and $\size(t)<\size(s)$.
\end{lemma}
\begin{proof}
   Let~$\k$ be $(4|\biQ|)^{|\biQ|}$.
   Let us fix a "context"~$d$ such that $\Prod(d)=P_0$.

   Consider now a "term"~$s$ of "height" larger than~$\k$ and some~$P_0\in\Reachable$.
   We aim at removing some piece of this "term" while achieving the conclusions of the lemma.

   For all "states"~$p\in P_0$, set~$\rho_p$ to be an optimal "run of~$\biA$ over~$s$ to~$p$", i.e.,
   \begin{itemize}
   \item if~$p\in\Q_\max$, then for all "runs~$\tau$ of~$\maxA$ over~$s$ to~$p$", $\weightrun(\tau)\leqslant\weightrun(\rho_p)$, and
   \item if~$p\in\Q_\min$, then for all "runs~$\tau$ of~$\minA$ over~$s$ to~$p$", $\weightrun(\rho_p)\leqslant \weightrun(\tau)$.
   \end{itemize}
   
   Since the longest branch of~$t$ has length greater than $2^{|\biQ|}2^{|\biQ|}|\biQ|^{|\biQ|}$, we can apply the 
   pigeonhole principle to the various ways to split this branch in two, and get a factorisation of~$s$ into
   \begin{align*}
       s = c\circ m \compterm s'\ ,
   \end{align*}
   in which~$c$ is a "context", $m$ is a non-empty "context", and $s'$ is a "term" such that
   \begin{itemize}
   \item $\Reach(s') = \Reach(m\compterm s')$; let~$R$ be this set;
   \item $\Prod(d\compterm c) = \Prod(d\compterm c\compterm m)$; let~$P$ be this set;
   \item for all~$p\in P_0$, there exists a state~$q_p\in\biQ$ such that~$\rho_p$ is decomposed into
	   a "run~$\tau_p$ over~$s'$ to~$q_p$", a "run~$\tau'_p$ over~$m$ from~$q_p$ to~$q_p$", and  a "run~$\tau''_p$
	   over~$s'$ to~$q_p$".
   \end{itemize}
  Let us define now our term~$t$ as:
   \begin{align*}
       t&=c \compterm s'\ .
   \end{align*}
   Since~$s=c\compterm m\compterm s'$, our new term~$t$ is nothing but~$s$
   in which the non-empty part corresponding to~$m$ has been removed. Hence~$\size(t)<\size(s)$.
   
   We shall prove now that "$t$ refines~$s$ for~$P_0$ with shift~$x$" where $x$
   is obtained by applying \cref{lemma:pumping} to~$P,R$ and~$m$.
   
   Let~$\rho$ be a "run of~$\maxA$ over~$s$ to state~$p$" for some~$p\in P_0$.
   We know that the run~$\rho_p$ as defined above is such that~$\weightrun(\rho)\leqslant\weightrun(\rho_p)$.
   Finally, let~$\rho'$ be the "run over~$t=c\compterm s'$ to~$p$"
   obtained by gluing~$\tau_p$ and~$\tau''_p$ together. 
   We have:
   \begin{align*}
   \weightrun(\rho)&\leqslant\weightrun(\rho_p) \tag{by optimality of~$\rho_p$}\\
   			&\leqslant\weightrun(\tau_p)+\weightrun(\tau'_p)+\weightrun(\tau''_p) \tag{decomposion of~$\rho_p$}\\
			&\leqslant\weightrun(\tau_p)+x+\weightrun(\tau''_p) \tag{by choice of~$x$ and \cref{lemma:pumping}}\\
			&\leqslant\weightrun(\rho')+x\ .\tag{definition of~$\rho'$}
   \end{align*}

   Hence, we have proved the first half of the definition of `"$t$ refines~$s$ for~$P_0$ with shift~$x$"'.
   The second half is symmetric. Overall, we conclude that "$t$ refines~$s$ for~$P_0$ with shift~$x$".
\end{proof}

Using iteratively the above \cref{lemma:always-refinable}, as long as the "height" of the term is larger than~$\k$, together with \cref{fact:shift-refine-reflexivity} and~\ref{fact:shift-refine-transitivity}, we obtain the following corollary.
\begin{corollary}\label{corollary:always-refinable}
There exists a computable~$\k\in\Nats$ such that for all $P\in\Reachable$ and all "terms"~$s$
there exists effectively a "term"~$t$ of "height" at most~$\k$ which "refines $s$ for~$P$ with some shift".
\end{corollary}

\subsection{The construction}
\label{section:construction}

\AP We are now ready to construct our separating automaton~$\sepA$. It is defined as follows:\phantomintro{\automaton_{\sep}}
\begin{align*}
 \reintro*\sepA=(\Q_\sep,\A,\Final_\sep,\Transitions_\sep,\weight_\sep)\ ,
\end{align*}
in which the set of states is%
\phantomintro{\Q_{\sep}}
\begin{multline*}
  \reintro*\Q_\sep= \{(R,P,t)\mid R\in\Reachable,\ P\in\Productive,\\
  			\ t\in\Terms,\ \Reach(t)=R,\ \height(t)\leqslant \k\}\ ,
\end{multline*}
(where~$\k$ is the constant from \cref{corollary:always-refinable}),
the \AP "final states", together with their "weight", are\phantomintro{\Final_{\sep}}\phantomintro{\weight_{\sep}}
\begin{align*}
  \reintro*\Final_\sep&=\{(R,P,t)\in\Q_\sep\mid P=\biFinal,\ R\cap\biFinal \neq\emptyset\}
  &\text{with}\quad\reintro*\weight_\sep(R,\biFinal,t)=\sem\maxA(t)\ ,
\end{align*}
and the "transition relation" and the weights are defined as follows.
\AP
For a "letter~$a$ of rank~$n$", there is a transition of the form%
\phantomintro{\Transitions_{\sep}}
\begin{align*}
  \delta=((R_0,P_0,t_0),\dots,(R_{n-1},P_{n-1},t_{n-1}),a,(R,P,t))\in \reintro*\Transitions_\sep\quad\text{with}\quad\reintro*\weight_\sep(\delta)=x\ ,
\end{align*}
whenever
\begin{itemize}
\item $((R_0,P_0),\dots,(R_{n-1},P_{n-1}),a,(R,P))$ is a "transition" of $\Transitions_\pro$. 
\item\AP  $(t,x)=\sr_P(a(t_0),\dots,a(t_{n-1}))$, where $\sr$ is a map of the following form:
	\phantomintro\sr
	\begin{align*}
	\reintro*\sr_P\colon \Terms &\longrightarrow \Terms \times \Reals\\
	s&\longmapsto
			(t,x)\qquad \text{such that~"$t$ refines~$s$ for~$P$ with shift~$x$".}
	\end{align*}
	(Such a map exists thanks to~\cref{corollary:always-refinable}.)
\end{itemize}

Let us first note:
\begin{lemma}\label{lemma:sep-total}
  For all~$P\in\Productive$ and all "terms"~$s$, there exists exactly one "run of~$\sepA$ over~$s$ to a state of
  the form~$(R,P,t)$".
\end{lemma}
\begin{proof}
Indeed, we have seen in \cref{lemma:lookahead} that $\sepA$ is "unambiguous@NFA" on its first two components. Then the third component is computed in a bottom-up deterministic manner. Furthermore, it is easy to show by induction that on every input "term" there is an accepting run.
\end{proof}

\begin{lemma}\label{lemma:sep-IH}
  Let~$\rho$ be a "run of $\sepA$ over~$s$ to~$(R,P,t)$", then 
  "$t$ refines~$s$ for~$P$ with shift~$\weightrun(\rho)$".
\end{lemma}
\begin{proof}
  The proof is by induction on $\height(s)$. Assume~$s$ of the form~$a(s_0,\dots,s_{n-1})$.
  Let~$\rho$ be the "run of~$\sepA$ over~$s$ to~$(R,P,t)$", let~$\delta=((R_0,P_0,t_0),\dots,(R_1,P_1,t_1),a,(R,P,t)$ be the transition assumed by~$\rho$ at the "root".
  Let~$\rho_i$ be the run~$\rho$ restricted to the subterm~$s_i$.
  By induction hypothesis, "$t_i$ refines~$s_i$ for~$P_i$ with shift~$\weightrun(\rho_i)$".
  By \cref{fact:shift-refine-transitivity}, 
    "$a(t_0,\dots,t_{n-1})$ refines~$s$ for~$P$ with shift~$\weightrun(\rho_0)+\cdots+\weightrun(\rho_{n-1})$".
  By definition of~$\weight_\sep$, "$t$ refines~$a(t_0,\dots,t_{n-1})$ with shift~$\weight_\sep(\delta)$".
  By \cref{fact:shift-refine-transitivity}, we obtain that "$t$ refines $s$ with shift~$\weightrun(\rho_0)+\cdots+\weightrun(\rho_{n-1})+\weight_\sep(\delta) = \weightrun(\rho)$".
\end{proof}

We can now provide the concluding lemma of the proof of \cref{theorem:main}.
\begin{lemma}\label{lemma:correctness}
  $\sem\maxA\leqslant\sem\sepA\leqslant\sem\minA$\ .
\end{lemma}
\begin{proof}
  Let~$s$ be a "term". By \cref{lemma:sep-total},
  there exists one and exactly one "run~$\rho_\sep$ of~$\sepA$ over~$s$ to a state of the form~$(R,\biFinal,t)$".
  By \cref{lemma:sep-IH}, "$t$ refines~$s$ for~$\biFinal$ with shift~$\weightrun(\rho_\sep)$".
  Note that in this case~$R=\Reach(s)=\Reach(t)$.

  Two cases can occur. If~$(R,\biFinal,t)$ is not "final". In this case, there is no "accepting run
  of~$\sepA$ over~$s$", and~$\sem\sepA(s)=\bot$.
  However,  $(R,\biFinal,t)\not\in\Final_\sep$ means $\Reach(t)\cap \biFinal=\emptyset$,
  hence~$\Reach(s)\cap\biFinal=\emptyset$. Thus $\sem\maxA(s)=\sem\minA(s)=\bot$.
  We indeed have~$\sem\maxA(s)\leqslant\sem\sepA(s)\leqslant\sem\minA(s)$.
  
  Otherwise, $(R,\biFinal,t)$ is "final", i.e. $R\cap\biFinal\neq\emptyset$. Assume for instance
  that there is some~$R\cap\biFinal\cap\Q_\max\neq\emptyset$ (it would be the same for~$\Q_\min$).
  This means that~$\sem\maxA(s)\neq\bot$.
  Since~$\sem\minA\geqslant\sem\maxA$, this implies also $\sem\minA(s)\neq\bot$.
  
  Let now~$\rho$ be an "accepting run of~$\maxA$ over~$s$" of maximal value, and let~$p$ be its "root" state.
  Since "$t$ refines~$s$ for~$\biFinal$ with shift~$\weightrun(\rho_\sep)$", and~$p\in\biFinal$,
  there exists a "run $\rho'$  over~$t$ to state~$p$"
  such that~$\weightrun(\rho)\leqslant\weightrun(\rho')+\weightrun(\rho_\sep)$.
  Hence,
  \begin{align*}
  \sem\maxA(s)&=\weightacc(\rho)\\
						&=\weightrun(\rho) + \weight_\max(p)\\
						&\leqslant \weightrun(\rho') + \weight_\max(p) + \weight(\rho_\sep)\\
						&\leqslant \sem\maxA(t) + \weight(\rho_\sep)\\
						&=\sem\sepA(s)
  \end{align*}
  In a symmetrical way, we obtain:
  \begin{align*}
  \sem\sepA(s)	&= \sem\maxA(t) + \weight(\rho_\sep)\\
  						&\leqslant \sem\minA(t) + \weight(\rho_\sep)\tag{assumption $\sem\maxA\leqslant\sem\minA$}\\
						&\leqslant\sem\minA(s)\ . \tag{as for the other inequality}
   \end{align*}
   Hence, we have established the expected $\sem\maxA(s)\leqslant\sem\sepA(s)\leqslant\sem\minA(s)$.
\end{proof}

\section{Conclusion}
\label{section:conclusion}

We have established a separation result for tropical automata over trees. 

We did not pay attention to the complexity of our construction. In a practical implementation, an improvement can easily be made: It is not necessary to use all terms $t$ of "height" up to~$\k$ in $\Q_\sep$: it is sufficient to keep minimal ones for the "shift refine relation" for each $P\in\Productive$.
\Cref{corollary:always-refinable} gives a crude upper bound on a sufficient number of these minimal terms. It is not clear if the exact bound is significantly better.

More interestingly, our result is under the assumption that~$\sem{\automaton_{{\max}}}\leqslant\sem{\automaton_{{\min}}}$. A natural variant is to invert the inequality and ask whether separation is possible when  $\sem{\automaton_{{\min}}}\leqslant\sem{\automaton_{{\max}}}$. Some separation results exist in both variants (like interpolation results in logic), while some do not (separation of Büchi automata, or Lusin's theorem). For tropical automata, the assumption  $\sem{\automaton_{{\min}}}\leqslant\sem{\automaton_{{\max}}}$ would be more complicated than
the one in our theorem: it can be witnessed for instance by the fact that  it is not decidable anymore \cite{Krob94}.

Another interesting question is whether similar results can hold for weights other than reals. For instance here, our proof requires for the weights of our automata to be equipped with a monoid structure, that it is commutative (otherwise weighted tree automata are not well defined), a total order (for the hypotheses of \cref{theorem:main} to be meaningful), that the product be compatible with the order, and archimedianity (for the pumping argument in \cref{lemma:pumping} to hold). The usefulness of each of these assumption could be studied. What if the monoid is not commutative (over words)? What if the order is not total (and be, for instance a lattice)? What if the operation is not archimedian (and what does it mean in these more general cases)? And in all these situations, do we capture interesting forms of automata?

More generally, these results of separation are fascinating, and it would be interesting to understand at high level what kind of abstract arguments may explain them, or at least some of them, uniformly. 

\bibliography{biblio.bib}

\begin{thebibliography}{10}

\bibitem{ColombetDaviaud16}
Thomas Colcombet and Laure Daviaud.
\newblock Approximate comparison of functions computed by distance automata.
\newblock {\em Theory Comput. Syst.}, 58(4):579--613, 2016.
\newblock \href {https://doi.org/10.1007/s00224-015-9643-3}
  {\path{doi:10.1007/s00224-015-9643-3}}.

\bibitem{ColcombetDaviaudZuleger14}
Thomas Colcombet, Laure Daviaud, and Florian Zuleger.
\newblock Size-change abstraction and max-plus automata.
\newblock In {\em Mathematical Foundations of Computer Science 2014 - 39th
  International Symposium, {MFCS} 2014, Budapest, Hungary, August 25-29, 2014.
  Proceedings, Part {I}}, pages 208--219, 2014.
\newblock \href {https://doi.org/10.1007/978-3-662-44522-8\_18}
  {\path{doi:10.1007/978-3-662-44522-8\_18}}.

\bibitem{ColcombetDaviaudZuleger17}
Thomas Colcombet, Laure Daviaud, and Florian Zuleger.
\newblock Automata and program analysis.
\newblock In {\em Fundamentals of Computation Theory - 21st International
  Symposium, {FCT} 2017, Bordeaux, France, September 11-13, 2017, Proceedings},
  pages 3--10, 2017.
\newblock \href {https://doi.org/10.1007/978-3-662-55751-8\_1}
  {\path{doi:10.1007/978-3-662-55751-8\_1}}.

\bibitem{GaubertMairesse95}
St\'ephane Gaubert and Jean Mairesse.
\newblock chapter Task resource models and (max,+) automata, pages 133--144.

\bibitem{Hashiguchi82}
Kosaburo Hashiguchi.
\newblock Limitedness theorem on finite automata with distance functions.
\newblock {\em J. Comput. Syst. Sci.}, 24(2):233--244, 1982.
\newblock \href {https://doi.org/10.1016/0022-0000(82)90051-4}
  {\path{doi:10.1016/0022-0000(82)90051-4}}.

\bibitem{Hashiguchi88}
Kosaburo Hashiguchi.
\newblock Algorithms for determining relative star height and star height.
\newblock {\em Inf. Comput.}, 78(2):124--169, 1988.
\newblock \href {https://doi.org/10.1016/0890-5401(88)90033-8}
  {\path{doi:10.1016/0890-5401(88)90033-8}}.

\bibitem{KirstenLombardy09}
Daniel Kirsten and Sylvain Lombardy.
\newblock Deciding unambiguity and sequentiality of polynomially ambiguous
  min-plus automata.
\newblock In {\em 26th International Symposium on Theoretical Aspects of
  Computer Science, {STACS} 2009}, pages 589--600, 2009.
\newblock \href {https://doi.org/10.4230/LIPIcs.STACS.2009.1850}
  {\path{doi:10.4230/LIPIcs.STACS.2009.1850}}.

\bibitem{KlimannLombardyMairessePrieur04}
Ines Klimann, Sylvain Lombardy, Jean Mairesse, and Christophe Prieur.
\newblock Deciding unambiguity and sequentiality from a finitely ambiguous
  max-plus automaton.
\newblock {\em Theor. Comput. Sci.}, 327(3):349--373, 2004.
\newblock \href {https://doi.org/10.1016/j.tcs.2004.02.049}
  {\path{doi:10.1016/j.tcs.2004.02.049}}.

\bibitem{Krob94}
Daniel Krob.
\newblock The equality problem for rational series with multiplicities in the
  tropical semiring is undecidable.
\newblock {\em {IJAC}}, 4(3):405--426, 1994.
\newblock \href {https://doi.org/10.1142/S0218196794000063}
  {\path{doi:10.1142/S0218196794000063}}.

\bibitem{KuperbergBoom11}
Denis Kuperberg and Michael {Vanden Boom}.
\newblock Quasi-weak cost automata: {A} new variant of weakness.
\newblock In {\em {IARCS} Annual Conference on Foundations of Software
  Technology and Theoretical Computer Science, {FSTTCS} 2011}, pages 66--77,
  2011.
\newblock \href {https://doi.org/10.4230/LIPIcs.FSTTCS.2011.66}
  {\path{doi:10.4230/LIPIcs.FSTTCS.2011.66}}.

\bibitem{KuperbergBoom12}
Denis Kuperberg and Michael {Vanden Boom}.
\newblock On the expressive power of cost logics over infinite words.
\newblock In {\em Automata, Languages, and Programming - 39th International
  Colloquium, {ICALP} 2012, Warwick, UK, July 9-13, 2012, Proceedings, Part
  {II}}, pages 287--298, 2012.
\newblock \href {https://doi.org/10.1007/978-3-642-31585-5\_28}
  {\path{doi:10.1007/978-3-642-31585-5\_28}}.

\bibitem{Leung91}
Hing Leung.
\newblock Limitedness theorem on finite automata with distance functions: An
  algebraic proof.
\newblock {\em Theor. Comput. Sci.}, 81(1):137--145, 1991.
\newblock \href {https://doi.org/10.1016/0304-3975(91)90321-R}
  {\path{doi:10.1016/0304-3975(91)90321-R}}.

\bibitem{LombardyMairesse06}
Sylvain Lombardy and Jean Mairesse.
\newblock Series which are both max-plus and min-plus rational are unambiguous.
\newblock {\em {ITA}}, 40(1):1--14, 2006.
\newblock \href {https://doi.org/10.1051/ita:2005042}
  {\path{doi:10.1051/ita:2005042}}.

\bibitem{LombardyMairesse07}
Sylvain {Lombardy} and Jean {Mairesse}.
\newblock {Series which are both max-plus and min-plus rational are
  unambiguous}.
\newblock {\em arXiv e-prints}, page arXiv:0709.3257, Sep 2007.
\newblock \href {http://arxiv.org/abs/0709.3257} {\path{arXiv:0709.3257}}.

\bibitem{Rabin70}
Michael~O. Rabin.
\newblock Weakly definable relations and special automata.
\newblock {\em Mathematical Logic and Foundations of Set Theory}, pages 1--23,
  1970.

\bibitem{RestivoReutenauer84}
Antonio Restivo and Christophe Reutenauer.
\newblock On cancellation properties of languages which are support of rational
  series.
\newblock {\em Journal of Computer System Sciences}, 29:153--159, 1984.

\bibitem{Schutzenberger61b}
Marcel~Paul Sch{\"{u}}tzenberger.
\newblock On the definition of a family of automata.
\newblock {\em Information and Control}, 4(2-3):245--270, 1961.
\newblock \href {https://doi.org/10.1016/S0019-9958(61)80020-X}
  {\path{doi:10.1016/S0019-9958(61)80020-X}}.

\bibitem{Skrzypczak14}
Micha\l Skrzypczak.
\newblock Separation property for wb- and ws-regular languages.
\newblock {\em Logical Methods in Computer Science}, 10(1), 2014.
\newblock \href {https://doi.org/10.2168/LMCS-10(1:8)2014}
  {\path{doi:10.2168/LMCS-10(1:8)2014}}.

\bibitem{TherienWilke04}
Denis Th{\'{e}}rien and Thomas Wilke.
\newblock Nesting until and since in linear temporal logic.
\newblock {\em Theory Comput. Syst.}, 37(1):111--131, 2004.
\newblock \href {https://doi.org/10.1007/s00224-003-1109-3}
  {\path{doi:10.1007/s00224-003-1109-3}}.

\end{thebibliography}

\end{document}